\documentclass[11pt]{article}
\pdfoutput=1
\usepackage{amssymb,amsmath,amsfonts,amsthm}
\usepackage{color,graphics}
\usepackage{enumerate}
\usepackage{graphicx}
\usepackage{subfig}
\usepackage{calc}

\DeclareMathOperator{\fork}{fork}
\theoremstyle{plain}
\newtheorem{theorem}{Theorem}

\newtheorem{lemma}{Lemma}

\theoremstyle{definition}

\theoremstyle{remark}

\input epsf

\begin{document}
\title{Fork-forests  in bi-colored complete bipartite graphs}
\author{Maria Axenovich\thanks{Departments of Mathematics, Iowa State University, USA, 
   and Karlsruhe Institute of Technology,
   Email: {\tt axenovic@iastate.edu}.
   The research is supported in part by NSF-DMS grant 0901008 \hfill\break},
  Marcus Krug\thanks{Faculty of Informatics, Karlsruhe Institute of Technology, Germany,
   Email: {\tt marcus.krug@kit.edu}\hfill\break},
  Georg Osang\thanks{Karlsruhe Institute of Technology, Germany,
   Email: {\tt georg.osang@student.kit.edu}\hfill\break},
  and Ignaz Rutter\thanks{Faculty of Informatics, Karlsruhe Institute of Technology, Germany,
   Email: {\tt rutter@kit.edu}}}

\maketitle

\begin{abstract} 
  Motivated by the problem in \cite{tsz-elwob-09}, which studies the
  relative efficiency of propositional proof systems, $2$-edge
  colorings of complete bipartite graphs are investigated. It is shown
  that if the edges of $G=K_{n,n}$ are colored with black and white
  such that the number of black edges differs from the number of white
  edges by at most $1$, then there are at least $n(1-1/\sqrt{2})$
  vertex-disjoint forks with centers in the same partite set of $G$.
  Here, a fork is a graph formed by two adjacent edges of different
  colors. The bound is sharp. Moreover, an algorithm running
  in time $O(n^2 \log n \sqrt{n \alpha(n^2,n) \log n})$ and giving a
  largest such fork forest is found.
\end{abstract}
 {\bf Keywords:} {bi-colored star forests, balanced colorings, OBDD}\\

\section{Introduction}

Let $G=K_{n,n}$ with partite sets $X$ and $Y$  be edge colored with two colors. 
We  investigate a global unavoidable substructure in balanced colorings of 
$K_{n,n}$, i.e., those where the number of  edges of one color differs from the 
number of  edges of another color  by at most one. For a two-coloring $c$, of 
$E(G)$  we call a set $S$ of three vertices  a {\it fork} in $G$ centered in $X$ 
(or $Y$)  if  $S$ induces two edges of different  colors sharing a vertex in $X$ 
(or $Y$). A set of  vertex-disjoint forks all centered in $X$ (or $Y$)  is 
called a \emph{fork forest} centered at $X$ (or $Y$). The number of forks in a fork 
forest $F$ is the {\it size of a forest}, denoted $|F|$.  For a coloring $c$ of 
$G$  let $f(G,c)$ be the largest  size of a fork forest centered either at $X$ or at 
$Y$. Finally, let $f(n)$ be the minimum $f(G,c)$ taken over all  balanced 
colorings $c$  using two colors. Our main results is

\begin{theorem}
For any  $n>1$,   $   f(n) = (1- \frac{1}{\sqrt{2}} )n$. There is an algorithm 
finding a largest  fork forest centered at $X$ in any  two-colored  complete 
bipartite graph with partite  sets $X$ and $Y$ and  running in time $O(n^2 \log 
n \sqrt{n \alpha(n^2,n) \log n})$.
\end{theorem}

This problem has connections to both graph theory and theoretical computer 
science. On one hand, it belongs to a class of problems seeking 
color-alternating subgraphs or general large  unavoidable subgraphs 
in two-edge colored graphs,  see for example
\cite{ag-99, ad-10, ef-01, kw-99}. On the other hand, special subgraphs in 
bi-colored complete bipartite graphs correspond to substructures in binary
matrices.  Finding $f(G,c)$ allows to determine the corresponding parameter in matrices
and to prove the conjecture of Tveretina et al. in \cite{tsz-elwob-09}.  
In particular, our result $f(n)=  (1- \frac{1}{\sqrt{2}} )n$ is an improvement of 
the previously known bound $f(n) \geq \frac{1}{2}(1-\frac{1}{\sqrt{2}})n$ from 
\cite{tsz-elwob-09}.  This in turn improves the lower bound on resolution for
ordered binary decision diagrams.
 
We first prove the result for $f(n)$, and then reduce the problem of  finding 
largest fork forests to a problem of finding perfect matchings of minimum weight 
in edge-weighted  graphs. With a known algorithm for the latter problem, our 
main theorem follows. In all the calculations we omit  floors and ceilings when 
their usage is clear from the context.

\section{Bounds on $f(n)$}

  For the upper bound, take $G$ to be a two-colored $K_{n,n}$ with edges of one 
  color forming a graph isomorphic to   $K_{\frac{n}{\sqrt{2}}, 
  \frac{n}{\sqrt{2}}}$.\\

  For the lower bound consider a balanced coloring of edges of $K_{n,n}$ with 
  partite sets $X$ and $Y$ in black and white. Let $G_1$ be the graph formed by 
  the black edges,  and  let $G_2$ be such a graph formed by the white edges. Let $M$ be 
  a maximum  matching of $G_1$. By K\H{o}nig's theorem applied to $G_1$, there 
  is a vertex cover $S$, of $G_1$ such that $|S|=|M|$. We have that $ S\subseteq 
  V(M)$.  Let $A= V(M) \cap X$, $B=V(M) \cap Y$, $A' = A\cap S$, $B' = B\cap S$, 
  $A'' = A-A'$, $B''= B-B'$. Note that $|A'|=|B''|$ and $|A''|=|B'|$.  Then we 
  see that the vertex set $(X-A') \cup (Y-B')$ induces no edges in $G_1$, as 
  otherwise $S$ would not be a vertex cover. Assume, without loss of generality, 
  that $|A'|\geq |B'|$.  \\~\\

\begin{figure}[h!]
  \centering
    \includegraphics{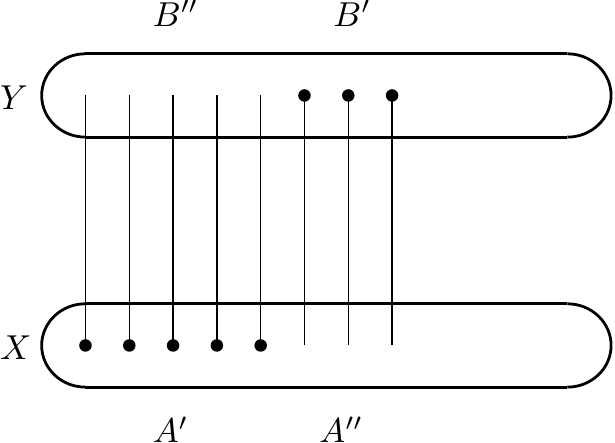}
  \caption{The matching edges $M$ and vertex cover $S$ of $G_1$, and the labelling
           of vertex sets introduced above.}
  \label{fig:sketch}
\end{figure}

  \noindent
  {\it Case 1}: $  |A'|\leq \frac{n}{\sqrt{2}}$.\\
  We have that $\frac{n^2}{2}= |E(G_1)| \leq
  n|A'| + (n-|A'|)|B'| \leq n|A'| + (n-|A'|)|A'| = 2n|A'|- |A'|^2$.
  So, from this we have that $|A'|\geq (1-\frac{1}{\sqrt{2}})n$.  Since
  $|X-A'| \geq (1-\frac{1}{\sqrt{2}})n$,
  there is a  fork forest centered at $B''$, using edges of $M$ and edges of 
  $G_2[B'', X-A']$  with $\min\{|B''|,|X-A'|\}\geq  (1-\frac{1}{\sqrt{2}})n$ forks.\\~\\

  \noindent {\it Case 2}: $ |A'|> \frac{n}{\sqrt{2}}$.\\
  Let $|A'| = \frac{n}{\sqrt{2}}+c$ for some positive $c$. We can
  assume that there is a matching $M'$ in $G_2$ of size at least
  $\frac{n}{\sqrt{2}}$, as otherwise Case 1 applies for $G_2$. 
  By counting, we can observe that at least $x := |M'| - |Y-B''| -
  |X-A'| \geq \frac{n}{\sqrt{2}} - 2(n-\frac{n}{\sqrt{2}}-c) =
  n(\frac{3}{\sqrt{2}}-2)+2c$ edges of  $M'$ have both endpoints in
  $A' \cup B''$. In the next two paragraphs we will show that 
  there are at least
  $\frac{x}{2}$  forks between $B''$ and $A'$ centered at $B''$:
  
  Consider  the union $G'=(M\cup M')[A' \cup B'']$, i.e., the black and white
  matching edges with one endpoint in $A'$ and the other in $B''$. There are $x$ edges 
  on $M'$ in this graph and each component  is  either an iterated even cycle or 
  a path ending with edges of $M$. It is easy to see that one could choose   at 
  least $\frac{k}{2}$ forks centered at $B''$ from  a component of $G'$ containing $k$ 
  edges of $M'$  and that is  either a path or a  cycle of length  divisible by 
  $4$. We also observe that one can choose $\frac{1}{2}(k_1+k_2)$ forks centered at $B''$ 
  from two cycles of $G'$ with $k_1$ and $k_2$ edges of $M'$,  where $k_1$ and 
  $k_2$ are odd,  by using a single  edge between these cycles and  additional 
  edges from the cycles.

\begin{figure}[h]
  \centering
  \subfloat[$P_5$]{\label{fig:match_a}{\includegraphics[height=0.4in]{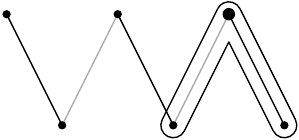}}}
    \hfill
  \subfloat[$C_8$]{\label{fig:match_b}{\includegraphics[height=0.4in]{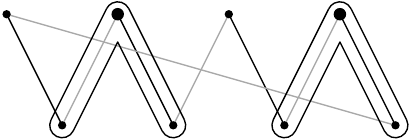}}}
    \hfill
  \subfloat[pair of $C_6$]{\label{fig:match_c}{\includegraphics[height=0.4in]{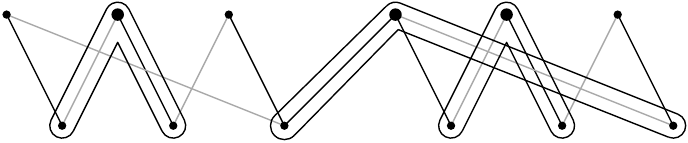}}}
    \hfill
  \subfloat[remaining $C_6$]{\label{fig:mtach_d}\includegraphics[height=0.4in]{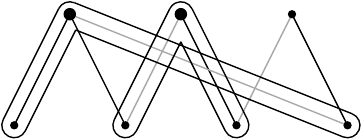}}
  \caption{Finding  forks in the components of $G'$, with one example for each type.
           White edges are drawn in light gray, the vertices belonging to $B''$ are
           positioned at the top. The outlined edges have been chosen to 
           be used in  forks. In the latter two cases this choice depends on the
           color of the single edge not contained in the cycles.}
  \label{fig:match}
\end{figure}

  So, we can pair up  all but at most one of the components of $G'$  that are 
  cycles of length $2$ modulo $4$. In the remaining  such component with $k$ 
  edges of  $M'$ we can choose $\frac{k+1}{2}$ forks centered at $B''$
  by using one additional edge going from the component into a previously
  unused vertex in $A'$ if available. If not, then all vertices in $A'$ have
  been used up from the previously chosen  forks,
  so we already have got $\lfloor\frac{1}{2}|A'|\rfloor \geq \lfloor\frac{n}{2\sqrt{2}}\rfloor$  forks.
  By combining the selected forks, we see that there are at least $\frac{x}{2}$ forks 
  centered in $B''$ and having leaves in $A'$.

  We observe that with each chosen  fork, at most two matching edges of $M$
  have become unavailable for later use, so there are at least $|B''|-x$ black
  matching edges with both endpoints in $A' \cup B''$  remaining. 
  These can be combined into  forks centered at $B''$ with nonedges
  leading into $X-A'$. This results in a total of
  $\frac{x}{2} + \min\{|B''|-x, |X-A'|\}$  forks. Since 

  \begin{eqnarray}
   \frac{x}{2} + \min\{|B''|-x, |X-A'|\}  
    & \ge& n(\frac{3}{2\sqrt{2}}-1)+c + \min\{n(2-\frac{2}{\sqrt{2}})-c, n(1-\frac{1}{\sqrt{2}})-c\} \nonumber \\
    &= & \min\{n(2-\frac{2}{\sqrt{2}}+\frac{3}{2\sqrt{2}}-1), n(1-\frac{1}{\sqrt{2}}+\frac{3}{2\sqrt{2}}-1)\} \nonumber \\
    &= & \min\{n(1-\frac{1}{2\sqrt{2}}), n(\frac{1}{2\sqrt{2}})\} = \frac{n}{2\sqrt{2}} \geq n(1-\frac{1}{\sqrt{2}}), \nonumber
  \end{eqnarray}
it follows that  $f(G,c) \geq  n(1-\frac{1}{\sqrt{2}})$.

\section{Algorithm}

We show that there is  an efficient algorithm for finding the largest fork 
forest centered at $X$ in $G$ by reducing this problem  to the problem of finding 
a perfect matching of minimum weight in an edge-weighted  graph~$G'$.    The 
case of a fork forest centered at~$Y$ is symmetric.\\

Informally, $G'$ is obtained from $G$ by first splitting each vertex of $X$ into
two adjacent vertices, with one of them being assigned the black edges incident
to the original vertex, and the other taking the white edges. Then all edges in $Y$ are added, and 
if $n$ is odd, one additional vertex is added  adjacent to all vertices of $Y$.\\
 
\begin{figure}[h]
  \centering
  \subfloat[original graph $G$]{\label{fig:transformation_a}{\includegraphics[width=0.3\textwidth]{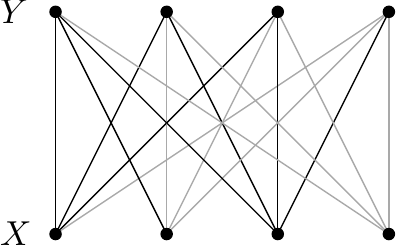}}}
    \hfill
  \subfloat[transformed graph $G'$]{\label{fig:transformation_b}\includegraphics[width=0.6\textwidth]{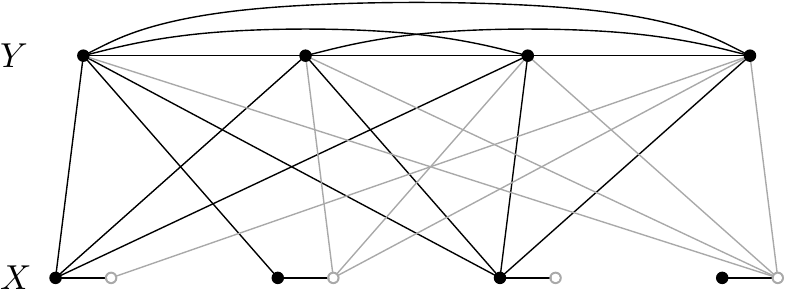}}
  \caption{A coloring of $G=K_{4,4}$ with white edges drawn in light gray,
           and its transformed version $G'$ on the right.
           For $x\in X$, vertices $x_b\in G'$ with the black edges incident to them
           are drawn in black, while $x_w \in G'$ with white incident edges are
           drawn in light gray.}
  \label{fig:transformation}
\end{figure}

\noindent
{\bf Construction}\\
For a $\{b,w\}$-coloring,  $c$, of    $G=K_{n,n}$ with partite sets $X$ and $Y$, 
   let $V(G')$ be a disjoint union  $Y'        \cup  \{ x_b:  x\in X\} 
\cup \{x_w: x\in X\}$, where $Y'=Y$ if $n$ is even and $Y'=Y\cup \{y\}$ if $n$ is odd.      Let $E(G')$ be the union of $ \{x_bx_w:  x\in 
X\}$, ~   $\{ yx_b:  c(yx)= b,x\in X, y \in Y \}$,   ~ $\{ yx_w:  c(yx)=w,x\in X, y \in Y\}$, and 
 all possible edges with endpoints in $Y'$.  Let   $\tau: E(G') \rightarrow 
\{0,1\}$ be such that $\tau(x_bx_w)=1$ for all $x\in X$,   and $\tau(e)=0$,  for all 
other edges.\\~\\

Further, if $M$ is a perfect matching in $G'$, denote by  $\fork(M)$ a fork forest in $G$ containing all 
forks on vertices $x, y, y'$ if $x_by\in M, x_wy'\in M$.  Recall that $|\fork(M)|$ is the number of forks in $\fork(M)$.


\begin{lemma}
If $M$ is a minimum weight perfect matching of~$(G', \tau)$ then $\fork(M)$ is a 
maximum fork forest of $(G, c)$ centered at $X$. 
\end{lemma}

\begin{proof}
Let $M$ be a minimum weight perfect matching of $(G', \tau)$. 
Note that  the weight of $M$ is equal to the number of edges  $x_bx_w \in E(M)$. 
We see that $x\not \in V(\fork(M))$ if and only if  $x_bx_w \in E(M)$, so the weight of $M$ 
is  $n- |\fork(M)|$. 

Assume that $\fork(M)$ is not a largest fork forest  of $(G, c)$ centered at $X$.
Then, for  a larger fork forest $F'$ of $(G, c)$ centered at $X$, let  $M'$  be 
a perfect matching of $G'$ that contains edges $x_by$ and $x_w y'$ if $x,y,y'$ 
induces a fork of $F'$, and  edge $x_bx_w$, otherwise.  Note that  one can 
always match vertices of $Y$ that are not in $F'$ with remaining vertices of $Y'$. This 
matching $M'$ has weight $n- |F'|< n-|\fork(M)|$,  a contradiction. \end{proof}

%
%
%
%

In~\cite{gt-fsagg-91} it is shown that the time complexity of finding the minimum weight matching in a graph with 
$n$ vertices,  $m$ edges, and edge-weights $0$ or $1$ is 
$O(\sqrt{n\alpha(m,n) \log n} m \log n),$ where~$\alpha$ denotes 
 the slowly growing inverse of the Ackermann function. Since~$G'$ contains at most~$3n+1$ vertices and~$\frac{3}{2}(n^2 + n)$ edges,  the minimum weight perfect matching problem for~$(G', w)$ can be solved 
 in~$O(n^2 \log n \sqrt{n \alpha(n^2,n) \log n})$ time.
Thus, the main theorem follows.\\~\\

\noindent
{\bf  Acknowledgements} The authors thank Olga Tveretina for bringing the problem to their attention.

\bibliographystyle{plain}
\bibliography{bi-forks}

\end{document}